\documentclass[journal]{IEEEtran}

\usepackage{cite}
\usepackage[cmex10]{amsmath}
\usepackage{amsthm,amssymb}
\usepackage{algorithmic}
\usepackage{array}
\usepackage{url}

\usepackage{graphicx}
\usepackage[percent]{overpic}

\usepackage[caption=false,font=footnotesize]{subfig}

\usepackage{url}   

\newtheorem{thm}{Theorem}

\newtheorem{lem}[thm]{Lemma}
\newtheorem{cor}[thm]{Corollary}

\newtheorem{example}{Example}

\newcommand{\mK}{\mathcal{K}}

\newcommand{\innerprod}[2]{  \left \langle #1 , #2  \right \rangle }

\newcommand{\dmk}{\mathrm{dist}(\mathcal{G}_m^{(K)},\neg\mathcal{G}_m^{(K)})}  
\newcommand{\udmk}{\underline{\mathrm{dist}}(\mathcal{G}_m^{(K)},\neg\mathcal{G}_m^{(K)})}

\newcommand{\WNR}{\text{WNR}}

\begin{document}

\title{Fingerprinting with Equiangular Tight Frames}
%
%
%

\author{Dustin~G.~Mixon,
        Christopher~J.~Quinn,~\IEEEmembership{Student~Member,~IEEE,}
        Negar~Kiyavash,~\IEEEmembership{Member,~IEEE,}
        and~Matthew~Fickus,~\IEEEmembership{Member,~IEEE}
\thanks{D.~G.~Mixon is with the Program in Applied and Computational Mathematics, Princeton University, Princeton, New Jersey 08544 USA (e-mail: dmixon@princeton.edu).}
\thanks{C.~J.~Quinn is with the Department of Electrical and Computer Engineering, University of Illinois, Urbana, Illinois 61801 USA (e-mail: quinn7@illinois.edu).}
\thanks{N.~Kiyavash is with the Department of Industrial and Enterprise Systems Engineering, University of Illinois, Urbana, Illinois 61801 USA (e-mail: kiyavash@illinois.edu).}
\thanks{M.~Fickus is with the Department of Mathematics and Statistics, Air Force Institute of Technology, Wright-Patterson AFB, Ohio 45433 USA (e-mail: matthew.fickus@afit.edu).}
\thanks{This work was presented in part at ICASSP 2011 and SPIE 2011. The authors thank the Air Force Summer Faculty Fellowship Program for making this collaboration possible.  This work was supported by NSF DMS 1042701, AFOSR F1ATA00083G004, AFOSR F1ATA00183G003, AFOSR FA9550-10-1-0345 and NRL N00173-09-1-G033.  Mixon was supported by the A.~B.~Krongard Fellowship.  Quinn was supported by the Department of Energy Computational Science Graduate Fellowship, which is provided under grant number DE-FG02-97ER25308.  The views expressed in this article are those of the authors and do not reflect the official policy or position of the United States Air Force, Department of Defense, or the U.S. Government.}
}

\maketitle



\begin{abstract}

Digital fingerprinting is a framework for marking media files, such as images, music, or movies, with user-specific signatures to deter illegal distribution.  
Multiple users can collude to produce a forgery that can potentially overcome a fingerprinting system.  
This paper proposes an equiangular tight frame fingerprint design which is robust to such collusion attacks.
We motivate this design by considering digital fingerprinting in terms of compressed sensing.  
The attack is modeled as linear averaging of multiple marked copies before adding a Gaussian noise vector.  
The content owner can then determine guilt by exploiting correlation between each user's fingerprint and the forged copy.
The worst-case error probability of this detection scheme is analyzed and bounded.
Simulation results demonstrate the average-case performance is similar to the performance of orthogonal and simplex fingerprint designs, while accommodating several times as many users.
\end{abstract}
 

%
\IEEEpeerreviewmaketitle

\section{Introduction}
%
%
%
%

\IEEEPARstart{D}{igital} media protection has become an important issue in recent years, as illegal distribution of licensed material has become increasingly prevalent.  
A number of methods have been proposed to restrict illegal distribution of media and ensure only licensed users are able to access it.  
One method involves cryptographic techniques, which encrypt the media before distribution.  
By doing this, only the users with appropriate licensed hardware or software have access; satellite TV and DVDs are two such examples.  
Unfortunately, cryptographic approaches are limited in that once the content is decrypted (legally or illegally), it can potentially be copied and distributed freely.

An alternate approach involves marking each copy of the media with a unique signature.  
The signature could be a change in the bit sequence of the digital file or some noise-like distortion of the media.  
The unique signatures are called \emph{fingerprints}, by analogy to the uniqueness of human fingerprints.  
With this approach, a licensed user could illegally distribute the file, only to be implicated by his fingerprint.
The potential for prosecution acts as a deterrent to unauthorized distribution.  

Fingerprinting can be an effective technique for inhibiting individual licensed users from distributing their copies of the media.  
However, fingerprinting systems are vulnerable when multiple users form a \emph{collusion} by combining their copies to create a forged copy.  
This attack can reduce and distort the colluders' individual fingerprints, making identification of any particular user difficult.  
Some examples of potential attacks involve comparing the bit sequences of different copies, averaging copies in the signal space, as well as introducing distortion (such as noise, rotation, or cropping).  

There are two principal approaches to designing fingerprints with robustness to collusions.
The first approach uses the \emph{marking} assumption~\cite{boneh1998collusion}: that the forged copy only differs from the colluders' copies where the colluders' copies are different (typically in the bit sequence).  
In many cases, this is a reasonable assumption because modifying other bits would likely render the file unusable (such as with software).  

Boneh and Shaw proposed the first known fingerprint design that uses the marking assumption to identify a member of the collusion with high probability~\cite{boneh1998collusion}. 
Boneh and Shaw's method incorporates the results of Chor et al., who investigated how to detect users who illegally share keys for encrypted material~\cite{chor2000tracing}.  
Schaathun later showed that the Boneh-Shaw scheme is more efficient than initially thought and proposed further improvements~\cite{schaathun2006boneh}.  
Tardos also proposed a method with significantly shorter codelengths than those of the Boneh-Shaw procedure~\cite{tardos2008optimal}.
Several recent works investigate the relationship between the fingerprinting problem and multiple access channels, and they calculate the capacity of a ``fingerprint channel''~\cite{somekh2005capacity, somekh2007achievable, anthapadmanabhan2008fingerprinting, moulin2008universal}. 
Barg and Kabatiansky also developed ``parent-identifying'' codes under a relaxation of the marking assumption~\cite{barg2003digital}, and there have been a number of other works developing special classes of binary fingerprinting codes, including~\cite{lin2009fingerprinting , cheng2011anti, cotrina2010family, boneh2010robust,koga2011digital,trappe2003anti}.

The second major approach uses the \emph{distortion} assumption.  
In this regime, fingerprints are noise-like distortions to the media in signal space.  
In order to preserve the overall quality of the media, limits are placed on the magnitude of this distortion. The content owner limits the power of the fingerprint he adds, and the collusion limits the power of the noise they add in their attack.  
When applying the distortion assumption, the literature typically assumes that the collusion linearly averages their individual copies to forge the host signal.
Also, while results in this domain tend to accommodate fewer users than those with the marking assumption, the distortion assumption enables a more natural embedding of the fingerprints, i.e., in the signal space.

Cox et al.~introduced one of the first robust fingerprint designs under the distortion assumption~\cite{cox1997secure}; the robustness was later analytically proven in~\cite{kilian1998resistance}.  
Ergun et al.~then showed that for any fingerprinting system, there is a tradeoff between the probabilities of successful detection and false positives imposed by a linear-average-plus-noise attack from sufficiently large collusions~\cite{ergun1999note}.  
Specific fingerprint designs were later studied, including orthogonal fingerprints~\cite{wang2005anti} and simplex fingerprints~\cite{kiyavash2009regular}.  
There are also some proposed methods motivated by CDMA techniques~\cite{hayashi2007collusion, li2005collusion}.  
Jourdas and Moulin demonstrated that a high rate can be achieved by embedding randomly concatenated, independent short codes~\cite{jourdas2009high}.
Kiyavash and Moulin derived a lower bound on the worst-case error probabilities for designs with equal-energy fingerprints~\cite{kiyavash2007sphere}.  
An error probability analysis of general fingerprint designs with unequal priors on user collusion is given in~\cite{dalkilic2010detection}.  
Some works have also investigated fingerprint performance under attack strategies other than linear averaging~\cite{wang2007capacity, ling2011novel} and under alternative noise models~\cite{kiyavash2009performance}.  

When working with distortion-type fingerprints, some embedding process is typically needed.  
This process is intended to make it difficult for the colluders to identify or distort fingerprints without significantly damaging the corresponding media in the process.  
For instance, if the identity basis is used as an orthogonal fingerprint design, then three or more colluders can easily identify and remove their fingerprints through comparisons.  
Indeed, different signal dimensions have different perceptual significance to the human user, and one should vary signal power strengths accordingly.  
There is a large body of work in this area known as \emph{watermarking}, and this literature is relevant since fingerprinting assigns a distinct watermark to each user.  
As an example, one method of fingerprint embedding is known as spread spectrum watermarking.  
Inspired by spread spectrum communications~\cite{pickholtz1982theory}, this technique rotates the fingerprint basis to be distributed across the perceptually significant dimensions of the signal~\cite{cox1997secure}.  
This makes fingerprint removal difficult while at the same time maintaining acceptable fidelity.  
For an overview of watermarking media, see~\cite{cox2002digital, hartung1999multimedia}.

The present paper proposes a fingerprint design under the distortion assumption. Specifically, we propose equiangular tight frames for fingerprint design and analyze their performance for the worst-case collusion. Moreover, through simulations, we show that these fingerprints perform comparably to orthogonal and simplex fingerprints on average, while accommodating several times as many users.  Li and Trappe \cite{li2005collusion} also used fingerprints satisfying the Welch bound, but did not use tight frames, designed the decoder to return the whole collusion, and did not perform worst case analysis.  Use of compressed sensing for fingerprint design was first suggested in  \cite{dai2008spherical} and \cite{varodayan2008collusion}, follow by \cite{pham2010compressive}.  However, \cite{pham2010compressive} and \cite{varodayan2008collusion} focused on detection schemes with Gaussian fingerprints.  Additionally, \cite{colbourn2010frameproof} examined combinatorial similarities of fingerprint codes and compressed sensing measurement matrices.  

We present the description of the fingerprinting problem in Section~\ref{sec:problem_setup}.
In Section~\ref{sec:ETF_fing_des}, we discuss this problem from a compressed sensing viewpoint and introduce the equiangular tight frame fingerprint design.
Using this design, we consider a detector which determines guilt for each user through binary hypothesis testing, using the correlation between the forged copy and the user's fingerprint as a test statistic.  
In Section~\ref{sec:error_analysis}, we derive bounds on the worst-case error probability for this detection scheme assuming a linear-average-plus-noise attack.  
Finally in Section~\ref{sec:simulations}, we provide simulations that demonstrate the average-case performance in comparison to orthogonal and simplex fingerprint designs.  



\section{Problem Setup} 
\label{sec:problem_setup}

In this section, we describe the fingerprinting problem and discuss the performance criteria we will use in this paper.
We start with the model we use for the fingerprinting and attack processes.

\subsection{Mathematical model}

A content owner has a host signal that he wishes to share, but he wants to mark it with fingerprints before distributing.
We view this host signal as a vector $s\in\mathbb{R}^N$, and the marked versions of this vector will be given to $M>N$ users.  
Specifically, the $m$th user is given 
\begin{equation} 
\label{eq:fingerprint_assignement}
x_m = s + f_m,
\end{equation} 
where $f_m\in\mathbb{R}^N$ denotes the $m$th fingerprint.  
We assume the fingerprints have equal energy:
\begin{equation}
\label{eq:energy}
\gamma^2:=\|f_m\|^2=ND_f, 
\end{equation}
that is, $D_f$ denotes the average energy per dimension of each fingerprint.

We wish to design the fingerprints $\{f_m\}_{m=1}^M$ to be robust to a linear averaging attack.
In particular, let $\mathcal{K}\subseteq\{1,\ldots,M\}$ denote a group of users who together forge a copy of the host signal.
Then their linear averaging attack is of the form
\begin{equation}
\label{eq:dm:attack1}
y=\sum_{k\in\mathcal{K}}\alpha_k(s+f_k)+\epsilon,\qquad \sum_{k\in\mathcal{K}}\alpha_k=1,
\end{equation}
where $\epsilon$ is a noise vector introduced by the colluders.  
We assume $\epsilon$ is Gaussian noise with mean zero and variance $N \sigma^2 $, that is, $\sigma^2$ is the noise power per dimension.  
The relative strength of the attack noise is measured as the \emph{watermark-to-noise} ratio (WNR): 
\begin{eqnarray} 
\label{eq:wnr}
\WNR := 10 \log_{10} \Big( \frac{N D_f}{N \sigma^2} \Big).
\end{eqnarray}  
This is analogous to signal-to-noise ratio.
See Figure~\ref{fig:attack_channel} for a schematic of this attack model.

\begin{figure}[t]
\centering
\includegraphics[width=.8\columnwidth]{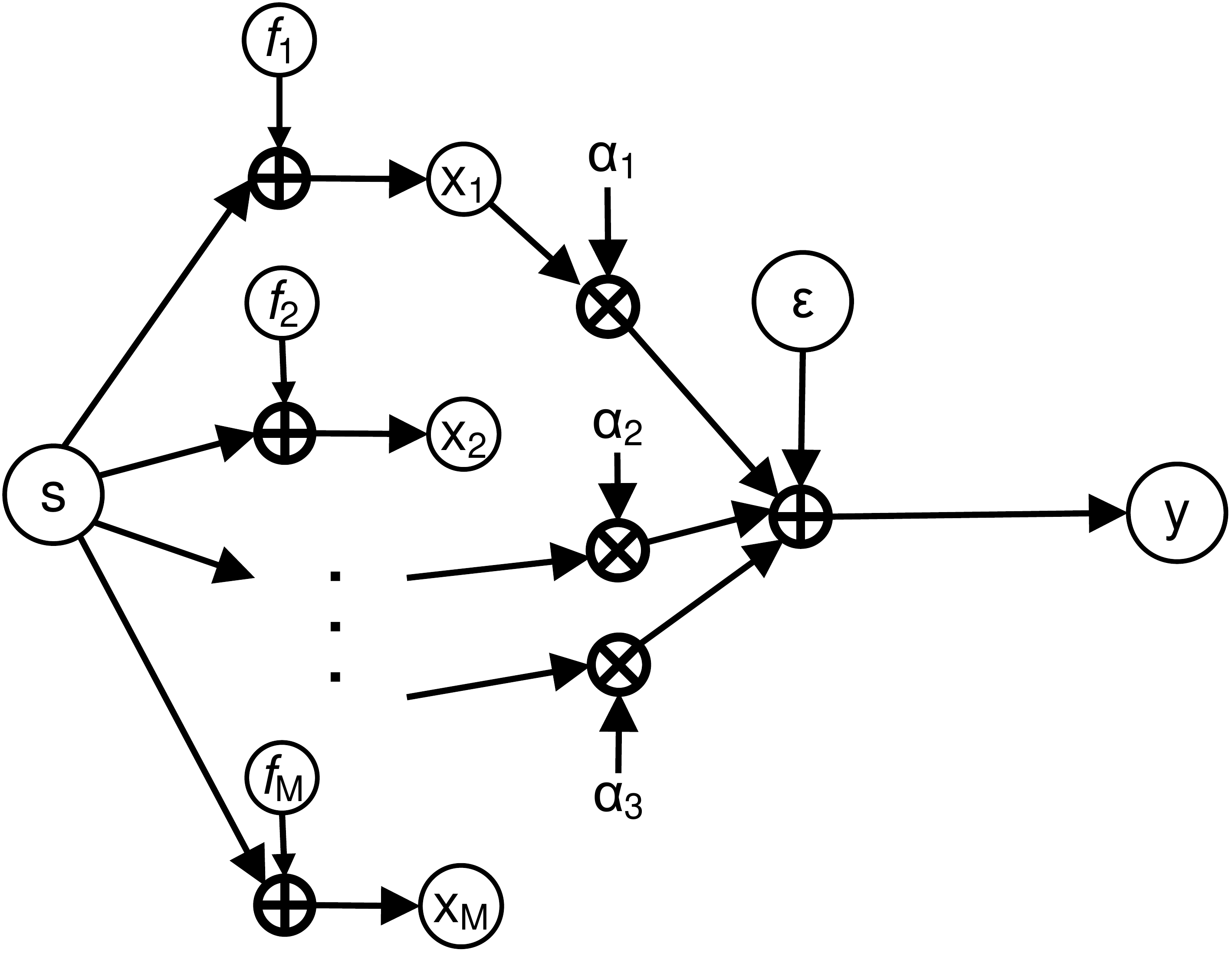}
\caption{The fingerprint assignment \eqref{eq:fingerprint_assignement} and forgery attack \eqref{eq:dm:attack1} processes.}
\label{fig:attack_channel}
\end{figure}

\subsection{Detection}

Certainly, the ultimate goal of the content owner is to detect every member in a forgery coalition.  
This can prove difficult in practice, though, particularly when some individuals contribute little to the forgery, with $\alpha_k \ll 1$.  
However, in the real world, if at least one colluder is caught, then other members could be identified through the legal process.
As such, we consider \emph{focused} detection, where a test statistic is computed for each user, and we perform a binary hypothesis test for each user to decide whether that particular user is guilty.  

With the cooperation of the content owner, the host signal can be subtracted from a forgery to isolate the fingerprint combination: 
\begin{equation*}
\label{eq:forgery} 
z := \sum_{k\in\mathcal{K}}\alpha_k f_k+\epsilon.
\end{equation*} 
We refer to $\sum_{k\in\mathcal{K}}\alpha_kf_k$ as the \emph{noiseless} fingerprint combination.  
The test statistic for each user $m$ is the normalized correlation function: 
\begin{equation} 
\label{eq:test_stat} 
\smash{T_m(z):=\frac{1}{\gamma^2}\langle z,f_m \rangle},
\end{equation} 
where $\gamma^2$ is the fingerprint energy \eqref{eq:energy}.  
For each user $m$, let $H_1(m)$ denote the guilty hypothesis $(m\in\mathcal{K})$ and $H_0(m)$ denote the innocent hypothesis $(m\not\in\mathcal{K})$.  
Letting $\tau$ denote a correlation threshold, we use the following detector:
\begin{equation} 
\label{eq:focused_detection}
\delta_m(\tau) := \left\{
   \begin{array}{cc}
     H_1(m), & T_m(z) \geq \tau,\\
     H_0(m), & T_m(z) < \tau.
   \end{array} \right.
\end{equation}  
To determine the effectiveness of our fingerprint design and focused detector, we will investigate the corresponding error probabilities.

\subsection{Error analysis}

Due in part to the noise that the coalition introduced to the forgery, there could be errors associated with our detection method.  
One type of error we can expect is the false-positive error, whose probability is denoted $\mathrm{P}_\mathrm{I}$, in which an innocent user $m$ ($m \notin \mathcal{K}$) is found guilty ($T_m(z) \geq \tau$).  
This could have significant ramifications in legal proceedings, so this error probability should be kept extremely low.  
The other error type is the false-negative error, whose probability is denoted $\mathrm{P}_\mathrm{II}$, in which a guilty user $(m \in \mathcal{K})$ is found innocent ($T_m(z) < \tau$).  
The probabilities of these two errors depend on the fingerprints $F=\{f_m\}_{m=1}^M$, the coalition $\mathcal{K}$, the weights $\alpha=\{\alpha_k\}_{k\in\mathcal{K}}$, the user $m$, and the threshold $\tau$:
\begin{align*}
\mathrm{P}_\mathrm{I}(F,m,\tau,\mathcal{K},\alpha)&:=\mathrm{Prob}[T_m(z)\geq\tau|H_0(m)],\\
\mathrm{P}_\mathrm{II}(F,m,\tau,\mathcal{K},\alpha)&:=\mathrm{Prob}[T_m(z)<\tau|H_1(m)].
\end{align*}  
We will characterize the \emph{worst-case} error probabilities over all possible coalitions and users.  

We first define the probability of a ``false alarm'': 
\begin{equation}
\label{eq:cq:P_fa}
\mathrm{P}_\mathrm{fa}(F,\tau,\mathcal{K},\alpha) := \max_{m\not\in\mathcal{K}}\mathrm{P}_\mathrm{I}(F,m,\tau,\mathcal{K},\alpha). 
\end{equation} 
This is the probability of wrongly accusing the innocent user who looks most guilty. 
Equivalently, this is the probability of accusing at least one innocent user.  
The worst-case type I error probability is given by
\begin{equation}
\label{eq:cq:worst_case_error_1}
\mathrm{P}_\mathrm{I}(F,\tau,\alpha):=\max_\mathcal{K}\mathrm{P}_\mathrm{fa}(F,\tau,\mathcal{K},\alpha).
\end{equation}  
Next, consider the probability of a ``miss'':
\begin{equation*}
\label{eq:cq:P_miss}
\mathrm{P}_\mathrm{m}(F,\tau,\mathcal{K},\alpha):=\min_{m\in\mathcal{K}}\mathrm{P}_\mathrm{II}(F,m,\tau,\mathcal{K},\alpha).
\end{equation*}  
This is the probability of not accusing the most vulnerable guilty user.  
Equivalently, this is the probability of not detecting any colluders.  
Note that this event is the opposite of detecting at least one colluder:
\begin{equation}
\label{eq:cq:P_d}
\mathrm{P}_\mathrm{d}(F,\tau,\mathcal{K},\alpha):=1-\mathrm{P}_\mathrm{m}(F,\tau,\mathcal{K},\alpha).
\end{equation}  
The worst-case type II error probability is given by
\begin{equation}
\label{eq:cq:worst_case_error_2}
\mathrm{P}_\mathrm{II}(F,\tau,\alpha):=\max_\mathcal{K}\mathrm{P}_\mathrm{m}(F,\tau,\mathcal{K},\alpha).
\end{equation}  
The \emph{worst-case} error probability is the maximum of the two error probabilities \eqref{eq:cq:worst_case_error_1} and \eqref{eq:cq:worst_case_error_2}: \begin{equation*}
\label{eq:minimax1} 
\mathrm{P}_\mathrm{e}(F,\tau,\alpha) :=\max \big\{\mathrm{P}_\mathrm{I}(F,\tau,\alpha), \mathrm{P}_\mathrm{II}(F,\tau,\alpha)\big\}. 
\end{equation*}  
The threshold parameter $\tau$ can be varied to minimize this quantity, yielding the minmax error probability:
\begin{equation}
\mathrm{P}_\mathrm{minmax}(F,\alpha) :=\min_{\tau} \mathrm{P}_\mathrm{e}(F,\tau,\alpha). \label{eq:setup:minmax_err}
\end{equation} 
In Section~\ref{sec:error_analysis}, we will analyze these error probabilities.  
We will also investigate average-case performance using simulations in Section~\ref{sec:simulations}.

\subsection{Geometric figure of merit} 
\label{sec:geom_fig_merit}

Related to the error probabilities is an important geometric figure of merit.
For each user $m$, consider the distance between two types of potential collusions: those of which $m$ is a member, and those of which $m$ is not.  
Intuitively, if every noiseless fingerprint combination involving $m$ is distant from every fingerprint combination not involving $m$, then even with moderate noise, there should be little ambiguity as to whether the $m$th user was involved or not.

To make this precise, for each user $m$, we define the ``guilty'' and ``not guilty'' sets of noiseless fingerprint combinations:
\begin{align*}
\mathcal{G}_m^{(K)}&:=\bigg\{\frac{1}{|\mathcal{K}|}\!\sum_{k\in\mathcal{K}}\!f_k:m\in\mathcal{K}\!\subseteq\!\{1,\ldots,M\},|\mathcal{K}|\!\leq\! K\bigg\},\\
\neg\mathcal{G}_m^{(K)}&:=\bigg\{\frac{1}{|\mathcal{K}|}\!\sum_{k\in\mathcal{K}}\!f_k:m\not\in\mathcal{K}\!\subseteq\!\{1,\ldots,M\},|\mathcal{K}|\!\leq\! K\bigg\}.
\end{align*}  
In words, $\mathcal{G}_m^{(K)}$ is the set of size-$K$ fingerprint combinations of equal weights ($\alpha_k = \frac{1}{K}$) which include $m$, while $\neg\mathcal{G}_m^{(K)}$ is the set of combinations which do not include $m$. 
Note that in our setup \eqref{eq:dm:attack1}, the $\alpha_k$'s were arbitrary nonnegative values bounded by one.  
We will show in Section~\ref{sec:error_analysis} that the best attack from the collusion's perspective uses equal weights ($\alpha_k = \frac{1}{K}$) so that no single colluder is particularly vulnerable.  
For this reason, we use equal weights to obtain bounds on the distance between these two sets, which we define to be
\begin{align}
\nonumber
&\dmk \\
\label{eqn:dist_guilt_notguilt}
&\qquad:= \min\{\|x-y\|_2:x\!\in\!\mathcal{G}_m^{(K)},y\!\in\!\neg\mathcal{G}_m^{(K)}\}.
\end{align}  
In Section~\ref{sec:ETF_fing_des}, we find a lower bound on this distance.

Another related parameter is the \emph{worst-case coherence}, which is the largest inner product between any two distinct fingerprints: 
\begin{equation} 
\label{eq:prb_setup:mu} 
\mu:=\max_{i\neq j}|\langle f_i,f_j \rangle|. 
\end{equation}  
Intuitively, we want this value to be small, as this would correspond to having the fingerprints spaced apart.  
In the following section, we discuss a way of designing fingerprints which we will later evaluate using the above criteria.

\section{ETF Fingerprint Design} \label{sec:ETF_fing_des}

In this section, we introduce a fingerprint design based on equiangular tight frames (ETFs).  
We will discuss some important properties of ETFs and use them to determine a lower bound on the distance \eqref{eqn:dist_guilt_notguilt} from our geometric figure of merit.  
But first, we consider the fingerprinting problem from a compressed sensing viewpoint.

\subsection{Compressed sensing viewpoint}

We wish to design fingerprints in a way that will enable us to identify the small group of users who take part in a collusion.
To do this, we consider a matrix-vector formulation of the attack \eqref{eq:dm:attack1} without noise, i.e., with $\epsilon=0$.
Specifically, let $F$ denote the $N \times M$ matrix whose columns are the fingerprints $\{f_m\}_{m=1}^M$, and let the $M \times 1$ vector $\alpha$ denote the weights used in the collusion's linear average.  
Note that $\alpha_{m}$ is zero if user $m$ is innocent, otherwise it's given by the corresponding coefficient in \eqref{eq:dm:attack1}.  
This gives 
\begin{equation*} 
\label{eq:attack:cs_view}
z = F\alpha. 
\end{equation*}  
Using this representation, the detection problem can be interpreted from a compressed sensing perspective.  
Namely, $\alpha$ is a $K$-sparse vector that we wish to recover.  
Under certain conditions on the matrix $F$, we may ``sense'' this vector with $N<M$ measurements in such a way that the $K$-sparse vector is recoverable:

\begin{thm}[$\!\!$\cite{candes2005decoding}]
Suppose an $N\times M$ matrix $F$ satisfies the restricted isometry property (RIP):
\begin{equation}
\label{eq.rip}
(1-\delta_{2K})\|\alpha\|_2^2\leq\|F\alpha\|_2^2\leq(1+\delta_{2K})\|\alpha\|_2^2
\end{equation}
for every $2K$-sparse vector $\alpha \in \mathbb{R}^M$, where $\delta_{2K} \leq \sqrt{2}-1$.
Then for every $K$-sparse vector $\alpha\in\mathbb{R}^M$,
\[ \alpha=\arg\min\|\hat{\alpha}\|_1\mbox{ subject to }F\hat{\alpha}=F\alpha. \]
\end{thm}  

Thus, if $F$ satisfies RIP \eqref{eq.rip}, we can recover the $K$-sparse vector $\alpha\in\mathbb{R}^M$ by linear programming.  
For the attack model with adversarial noise, $z = F\alpha + \epsilon$, if $F$ satisfies RIP \eqref{eq.rip}, linear programming will still produce an estimate $\hat{\alpha}$ of the sparse vector \cite{candes2006stable}.  
However, the distance between $\hat{\alpha}$ and $\alpha$ will be on the order of 10 times the size of the error $\epsilon$.  
Due to potential legal ramifications of false accusations, this order of error is not tolerable.  
Note that these methods (both for the noiseless and noisy cases) recover the entire vector $\alpha$; equivalently, they identify the entire collusion.
That said, we will investigate RIP matrices for fingerprint design, but to minimize false accusations, we will use focused detection \eqref{eq:focused_detection} to identify colluders.

\subsection{Geometric figure of merit for RIP matrices}

We now investigate how well RIP matrices perform with respect to our geometric figure of merit.
Without loss of generality, we assume the fingerprints are unit norm; since they have equal energy $\gamma^2$, the fingerprint combination $z$ can be normalized by $\gamma$ before the detection phase.   
With this in mind, we have the following a lower bound on the distance~\eqref{eqn:dist_guilt_notguilt} between the ``guilty'' and ``not guilty'' sets corresponding to any user $m$:

\begin{thm}
\label{thm.rip fingerprints}
Suppose fingerprints $F=[f_1,\ldots,f_M]$ satisfy the restricted isometry property \eqref{eq.rip}.
Then
\begin{equation} 
\label{lem:distbnd1}
\dmk\geq\sqrt{\frac{1-\delta_{2K}}{K(K-1)}}. 
\end{equation}
\end{thm}

\begin{proof}
Take $\mathcal{K},\mathcal{K}'\subseteq\{1,\ldots,M\}$ such that $|\mathcal{K}|,|\mathcal{K}'|\leq K$ and $m\in\mathcal{K}\setminus\mathcal{K}'$.
Then the left-hand inequality of the restricted isometry property \eqref{eq.rip} gives
\begin{align}
\nonumber
&\bigg\|\frac{1}{|\mathcal{K}|}\sum_{m\in\mathcal{K}}f_m-\frac{1}{|\mathcal{K}'|}\sum_{m\in\mathcal{K}'}f_m\bigg\|^2\\
\nonumber
&=\bigg\|\big(\tfrac{1}{|\mathcal{K}|}\!-\!\tfrac{1}{|\mathcal{K}'|}\big)\!\!\!\!\!\!\sum_{m\in\mathcal{K}\cap\mathcal{K}'}\!\!\!\!\!\!f_m+\tfrac{1}{|\mathcal{K}|}\!\!\!\!\!\!\sum_{m\in\mathcal{K}\setminus\mathcal{K}'}\!\!\!\!\!\!f_m-\tfrac{1}{|\mathcal{K}'|}\!\!\!\!\!\!\sum_{m\in\mathcal{K}'\setminus\mathcal{K}}\!\!\!\!\!\!f_m\bigg\|^2\\
\nonumber
&\geq(1-\delta_{|\mathcal{K}\cup\mathcal{K}'|})\bigg(|\mathcal{K}\cap\mathcal{K}'|\big(\tfrac{1}{|\mathcal{K}|}\!-\!\tfrac{1}{|\mathcal{K}'|}\big)^2+\tfrac{|\mathcal{K}\setminus\mathcal{K}'|}{|\mathcal{K}|^2}+\tfrac{|\mathcal{K}'\setminus\mathcal{K}|}{|\mathcal{K}'|^2}\bigg)\\
\label{eq.distance bound}
&=\frac{1-\delta_{|\mathcal{K}\cup\mathcal{K}'|}}{|\mathcal{K}||\mathcal{K}'|}\bigg(|\mathcal{K}|+|\mathcal{K}'|-2|\mathcal{K}\cap\mathcal{K}'|\bigg).
\end{align}
For a fixed $|\mathcal{K}|$, we will find a lower bound for
\begin{equation}
\label{eq.min value} \frac{1}{|\mathcal{K}|}\bigg(|\mathcal{K}|+|\mathcal{K}'|-2|\mathcal{K}\cap\mathcal{K}'|\bigg)=1+\frac{|\mathcal{K}|-2|\mathcal{K}\cap\mathcal{K}'|}{|\mathcal{K}'|}.
\end{equation}
Since we can have $|\mathcal{K}\cap\mathcal{K}'|>\frac{|\mathcal{K}|}{2}$, 
we know $\frac{|\mathcal{K}|-2|\mathcal{K}\cap\mathcal{K}'|}{|\mathcal{K}'|}<0$ when \eqref{eq.min value} is minimized.  
That said, $|\mathcal{K}'|$ must be as small as possible, that is, $|\mathcal{K}'|=|\mathcal{K}\cap\mathcal{K}'|$.
Thus, when \eqref{eq.min value} is minimized, we must have
\begin{equation*}
\frac{1}{|\mathcal{K}|}\bigg(|\mathcal{K}|+|\mathcal{K}'|-2|\mathcal{K}\cap\mathcal{K}'|\bigg)=\frac{|\mathcal{K}|}{|\mathcal{K}\cap\mathcal{K}'|}-1,
\end{equation*}
i.e., $|\mathcal{K}\cap\mathcal{K}'|$ must be as large as possible.
Since $m\in\mathcal{K}\setminus\mathcal{K}'$, we have $|\mathcal{K}\cap\mathcal{K}'|\leq|\mathcal{K}|-1$.
Therefore,
\begin{equation}
\label{eq.min value 2}
\frac{1}{|\mathcal{K}|}\bigg(|\mathcal{K}|+|\mathcal{K}'|-2|\mathcal{K}\cap\mathcal{K}'|\bigg)\geq\frac{1}{|\mathcal{K}|-1}.
\end{equation}
Substituting \eqref{eq.min value 2} into \eqref{eq.distance bound} gives 
\begin{equation*}
\bigg\|\tfrac{1}{|\mathcal{K}|}\!\!\sum_{m\in\mathcal{K}}\!\!f_m-\tfrac{1}{|\mathcal{K}'|}\!\!\sum_{m\in\mathcal{K}'}\!\!f_m\bigg\|^2
\geq\frac{1-\delta_{|\mathcal{K}\cup\mathcal{K}'|}}{|\mathcal{K}|(|\mathcal{K}|-1)}
\geq\frac{1-\delta_{2K}}{K(K-1)}.
\end{equation*}
Since this bound holds for every $m$, $\mathcal{K}$ and $\mathcal{K}'$ with $m\in\mathcal{K}\setminus\mathcal{K}'$, we have~\eqref{lem:distbnd1}.
\end{proof}

Note that \eqref{lem:distbnd1} depends on $\delta_{2K}$, and so it's natural to ask how small $\delta_{2K}$ can be for a given fingerprint design.
Also, which matrices even satisfy RIP~\eqref{eq.rip}? 
In general, we say $F$ is $(K,\delta)$-RIP if for every $K$-sparse vector $x$,
\begin{equation*}
(1-\delta)\|x\|_2^2\leq\|Fx\|_2^2\leq(1+\delta)\|x\|_2^2.
\end{equation*}
Proving that a given matrix $F$ is $(K,\delta)$-RIP involves calculating eigenvalues of all size-$K$ submatrices, which is computationally difficult for large matrices.  
The following lemma makes this explicit:

\begin{lem} 
\label{lem:delta_min_bnd}
The smallest $\delta$ for which $F$ is $(K,\delta)$-RIP is
\begin{equation}
\label{eq.delta min}
\delta_\mathrm{min}
:=\max_{\substack{\mathcal{K}\subseteq\{1,\ldots,M\}\\|\mathcal{K}|=K}}\|F_\mathcal{K}^*F_\mathcal{K}-\mathrm{I}_K\|_2.
\end{equation}
\end{lem}

\begin{proof}
We first note that $F$ being $(K,\delta)$-RIP trivially implies that $F$ is $(K,\delta+\epsilon)$-RIP for every $\epsilon>0$.
It therefore suffices to show that (i) $F$ is $(K,\delta_\mathrm{min})$-RIP, and (ii) $F$ is not $(K,\delta)$-RIP for any $\delta<\delta_\mathrm{min}$.
To this end, pick some $K$-sparse vector $x$.
To prove (i), we need to show that
\begin{equation}
\label{eq.rip min}
(1-\delta_\mathrm{min})\|x\|^2
\leq\|Fx\|^2
\leq(1+\delta_\mathrm{min})\|x\|^2.
\end{equation}
Let $\mathcal{K}\subseteq\{1,\dots,M\}$ be the size-$K$ support of $x$, and let $x_\mathcal{K}$ be the corresponding subvector.
Then rearranging \eqref{eq.rip min} gives
\begin{align}
\nonumber
\delta_\mathrm{min}
\geq\Big|\tfrac{\|Fx\|^2}{\|x\|^2}-1\Big|
&=\Big|\tfrac{\langle F_\mathcal{K}x_\mathcal{K},F_\mathcal{K}x_\mathcal{K}\rangle-\langle x_\mathcal{K},x_\mathcal{K}\rangle}{\|x_\mathcal{K}\|^2}\Big|\\
\label{eq.delta min bound}
&=\Big|\Big\langle \tfrac{x_\mathcal{K}}{\|x_\mathcal{K}\|},(F_\mathcal{K}^*F_\mathcal{K}-\mathrm{I}_K)\tfrac{x_\mathcal{K}}{\|x_\mathcal{K}\|} \Big\rangle\Big|.
\end{align}
Since, by definition, $\delta_\mathrm{min}$ maximizes \eqref{eq.delta min bound} over all supports $\mathcal{K}$ and entry values $x_\mathcal{K}$, the inequality necessarily holds; that is, $F$ is necessarily $(K,\delta_\mathrm{min})$-RIP.
Furthermore, equality is achieved by the support $\mathcal{K}$ which maximizes \eqref{eq.delta min} and the eigenvector $x_\mathcal{K}$ corresponding to the largest eigenvalue of $F_\mathcal{K}^*F_\mathcal{K}-\mathrm{I}_K$; this proves (ii).
\end{proof}

Since calculating eigenvalues for all size-$K$ submatrices is computationally difficult, the Gershgorin circle theorem is often used to obtain a coarse bound in demonstrating RIP:

\begin{thm}[Gershgorin circle theorem \cite{varga2004geršgorin}] 
\label{thm:gersh}
Take a $K\times K$ matrix $A$.
Then for each eigenvalue $\lambda$ of $A$, there exists an $i\in\{1,\ldots,K\}$ such that $\lambda$ lies in the complex disk centered at $A_{ii}$ of radius $\sum_{j\neq i}|A_{ij}|$.
\end{thm}  

This leads to the following well-known bound on $\delta_{2K}$ in terms of worst-case coherence:

\begin{lem} \label{lem:bnd_delta2K}
Given a matrix $F$ with unit-norm columns, then 
\begin{equation} 
\label{eqn:lem_bnd_delta2K}
\delta_{2K}\leq(2K-1)\mu,  
\end{equation} 
where $\mu$ is the worst-case coherence \eqref{eq:prb_setup:mu}.
\end{lem}

\begin{proof}   
From Lemma~\ref{lem:delta_min_bnd}, the optimal $\delta_{2K}$ is
\begin{equation*}
\label{eq.delta} 
\delta_{2K}:=\max_{\mathcal{K}\subseteq\{1,\ldots,M\}\atop|\mathcal{K}|=2K}\|F_\mathcal{K}^*F_\mathcal{K}-\mathrm{I}_K\|_2.
\end{equation*}
In words, given an ensemble of fingerprints, we consider each subcollection of size $2K$, subtract the identity from their $2K\times 2K$ Gram matrix and calculate the largest eigenvalue; the largest eigenvalue we find over all subcollections is what we call $\delta_{2K}$.  
Since the fingerprints have unit norm, every Gram matrix of $2K$ fingerprints will have ones on the diagonal, and the absolute values of the off-diagonal entries will be no more than $\mu$.
Therefore, by applying the Gershgorin circle theorem (Theorem~\ref{thm:gersh}) to $\delta_{2K}$, we obtain \eqref{eqn:lem_bnd_delta2K}.
\end{proof}
  
Combining Theorem~\ref{thm.rip fingerprints} and Lemma~\ref{lem:bnd_delta2K} yields a coherence-based lower bound on the distance between the ``guilty'' and ``not guilty'' sets corresponding to any user $m$:
\begin{thm}
\label{thm:etf_fingerprints}
Suppose fingerprints $F=[f_1,\ldots,f_M]$ are unit-norm with worst-case coherence $\mu$.  
Then
\begin{equation}
\label{eq:dm:bound_dist_guilty_notguilty_set}
\dmk\geq\sqrt{\frac{1-(2K-1)\mu}{K(K-1)}}. 
\end{equation} 
\end{thm}
 
We would like $\mu$ to be small, so that the lower bound \eqref{eq:dm:bound_dist_guilty_notguilty_set} is as large as possible.  
But for a fixed $N$ and $M$, the worst-case coherence of unit-norm fingerprints cannot be arbitrarily small; it necessarily satisfies the Welch bound \cite{welch1974lower}: 
\begin{equation*} 
\mu\geq\sqrt{\frac{M-N}{N(M-1)}}. 
\end{equation*} 
Equality in the Welch bound occurs precisely in the well-studied case where the fingerprints form an equiangular tight frame (ETF).  
An \emph{equiangular tight frame} is a $N \times M$ matrix which has orthogonal rows of equal norm and unit-norm columns whose inner products have equal magnitude~\cite{strohmer2003grassmannian}.

One type of ETF has already been proposed for fingerprint design: the simplex~\cite{kiyavash2009regular}. 
The simplex is an ETF with $M=N+1$ and $\mu=\frac{1}{N}$.  
In fact, \cite{kiyavash2009regular} gives a derivation for the value of the distance \eqref{eqn:dist_guilt_notguilt} in this case:
\begin{equation} 
\label{eq.simplex distance} 
\mathrm{dist}(\mathcal{G}_m^{(K)},\neg\mathcal{G}_m^{(K)})
=\sqrt{\frac{1}{K(K-1)}\frac{M}{M-1}}.
\end{equation}  
The bound \eqref{eq:dm:bound_dist_guilty_notguilty_set} is lower than \eqref{eq.simplex distance} by a factor of $\sqrt{1 - \frac{2K}{N+1}}$, and for practical cases in which $K \ll N$, they are particularly close.   
Overall, ETF fingerprint design is a natural generalization of the provably optimal simplex design of~\cite{kiyavash2009regular}.

\subsection{Existence of RIP matrices}

\begin{figure*}[t]
\begin{equation*} 
F=\frac1{\sqrt{3}}\left[\begin{array}{cccccccccccccccc}
+&-&+&-&+&-&+&-\\
+&+&-&-&&&&&+&-&+&-\\
+&-&-&+&&&&&&&&&+&-&+&-\\
&&&&+&+&-&-&+&+&-&-\\
&&&&+&-&-&+&&&&&+&+&-&-
\\&&&&&&&&+&-&-&+&+&-&-&+  
\end{array}\right].
\end{equation*}
\caption{The equiangular tight frame constructed in Example~\ref{ex:ETF}.}
\label{fig:etf_construction}
\end{figure*}

We now discuss the existence of RIP matrices as well as the construction of a specific class of ETFs.  
Given the RIP definition \eqref{eq.rip}, it might not be clear how many matrices satisfy this condition.  
Surprisingly, RIP matrices are abundant; for any $\delta_{2K}$, there exists a constant $C$ such that if $N\geq CK\log(M/K)$, then matrices of Gaussian or Bernoulli ($\pm1$) entries satisfy RIP with high probability.  
However, as mentioned in the previous section, checking if a given matrix satisfies RIP is computationally difficult.  

Fortunately, there are some deterministic constructions of RIP matrices, such as \cite{devore2007deterministic,fickus2010steiner}.  
However the performance, measured by how large $K$ can be, is not as good as the random constructions; the random constructions only require $N=\Omega(K\log^a M)$, while the deterministic constructions require $N=\Omega(K^2)$.
The specific class of ETFs additionally requires $M\leq N^2$ \cite{strohmer2003grassmannian}.

Whether $N$ scales as $K$ versus $K^2$ is an important distinction between random and deterministic RIP matrices in the compressed sensing community.  However, this difference offers no advantage for fingerprinting.  Ergun et al.~showed that for any fingerprinting system, a collusion of $K=O(\sqrt{N/\ln N})$ is sufficient to overcome the fingerprints \cite{ergun1999note}.  This means with such a $K$, the detector cannot, with high probability, identify any attacker without incurring a significant false-alarm probability $\mathrm{P}_\mathrm{fa}$ \eqref{eq:cq:P_fa}.  This constraint is more restrictive than  deterministic ETF constructions, as $\sqrt{N/\ln N} < \sqrt{N} \ll N / \log^a M$.  Consequently, random RIP constructions are no better for fingerprint design than deterministic constructions.

Ergun's bound indicates that with a large enough $K$, colluders can overcome any fingerprinting system and render our detector unreliable.  We now show that if $K$ is sufficiently large, the colluders can \emph{exactly} recover the original signal $s$:

\begin{lem}
Suppose the real equal-norm fingerprints $\{f_k\}_{k\in\mathcal{K}}$ do not lie in a common hyperplane.
Then $s$ is the unique minimizer of
\begin{equation*}
g(x):=\sum_{k\in\mathcal{K}}\bigg(\|x-(s+f_k)\|^2-\frac{1}{|\mathcal{K}|}\sum_{k'\in\mathcal{K}}\|x-(s+f_{k'})\|^2\bigg)^2.
\end{equation*}
\end{lem}

\begin{proof}
Note that $g(x)\geq0$, with equality precisely when $\|x-(s+f_k)\|^2$ is constant over $k\in\mathcal{K}$.
Since the $f_k$'s have equal norm, $g(s)=0$. To show that this minimum is unique, suppose $g(x)=0$ for some $x \neq s$. This implies that the values $\| x - (s+f_k)\|^2$ are constant, with each being equal to their average.  Moreover, since 
\begin{equation*}
\|x-(s+f_k)\|^2=\|x-s\|^2-2\langle x-s,f_k\rangle+\|f_k\|^2,
\end{equation*} we have that $\langle x-s, f_k \rangle$ is constant over $k \in \mathcal{K}$, contradicting the assumption that the fingerprints $\{f_k\}_{k\in\mathcal{K}}$ do not lie in a common hyperplane.
%
%
%
\end{proof}

\subsection{Construction of ETFs}

Having established that deterministic RIP constructions are just as good as random constructions for fingerprint design, we now consider a particular method for constructing ETFs.
Note that ETFs are notoriously difficult to construct in general, but a relatively simple approach was recently introduced in~\cite{fickus2010steiner}.  
The approach uses a tensor-like combination of a Steiner system's adjacency matrix and a regular simplex, and it is general enough to construct infinite families of ETFs.   
We illustrate the construction with an example:

\begin{example}[see \cite{fickus2010steiner}]
\label{ex:ETF} 

To construct an ETF, we will use a simple class of Steiner systems, namely $(2,2,v)$-Steiner systems, and Hadamard matrices with $v$ rows.  
In particular, $(2,2,v)$-Steiner systems can be thought of as all possible pairs of a $v$-element set \cite{van2001course}, while a Hadamard matrix is a square matrix of $\pm1$'s with orthogonal rows \cite{van2001course}.  

In this example, we take $v=4$.  
The adjacency matrix of the $(2,2,4)$-Steiner system has $\binom{4}{2} = 6$ rows, each indicating a distinct pair from a size-$4$ set:
\begin{equation}
\label{eq:steiner_mat}
A=\begin{bmatrix}+&+&&\\+&&+&\\+&&&+\\&+&+&\\&+&&+\\&&+&+\end{bmatrix}.
\end{equation}  
To be clear, we use ``$+/-$'' to represent $\pm 1$ and an empty space to represent a $0$ value.  
Also, one example of a $4 \times 4 $ Hadamard matrix is
\begin{equation} 
\label{eq:had_mat}
H=\begin{bmatrix}+&+&+&+\\+&-&+&-\\+&+&-&-\\+&-&-&+\end{bmatrix}.
\end{equation} 
We now build an ETF by replacing each of the ``$+$'' entries in the adjacency matrix \eqref{eq:steiner_mat} with a row from the Hadamard matrix \eqref{eq:had_mat}.  
This is done in such a way that for each column of $A$, distinct rows of $H$ are used; in this example, we use the second, third, and fourth rows of $H$.  
After performing this procedure, we normalize the columns, and the result is a real ETF $F$ of dimension $N=6$ with $M=16$ fingerprints (see Figure~\ref{fig:etf_construction}).  
\end{example}  

We will use this method of ETF construction for simulations in Section~\ref{sec:simulations}.  Specifically, our simulations will let the total number of fingerprints $M$ range from $2N$ to $7N$, a significant increase from the conventional $N+1$ simplex and $N$ orthogonal fingerprints.  Note, however, that unlike simplex and orthogonal fingerprints, there are limitations to the dimensions that admit real ETFs.  For instance, the size of a Hadamard matrix is necessarily a multiple of $4$, and the existence of Steiner systems is rather sporadic.  Fortunately, identifying Steiner systems is an active area of research, with multiple infinite families already characterized and alternative construction methods developed \cite{fickus2010steiner, ccrwest}.

\section{Error Analysis} \label{sec:error_analysis}

\subsection{Analysis of type I and type II errors}

We now investigate the worst case errors involved with using ETF fingerprint design and focused correlation detection under linear averaging attacks.  Recall that the worst-case type I error  probability \eqref{eq:cq:worst_case_error_1} is the probability of falsely accusing the most guilty-looking innocent user.  Also recall that the worst-case type II error  probability \eqref{eq:cq:worst_case_error_2} is the probability of not accusing the most vulnerable guilty user.  

\begin{thm}
\label{thm:cq:linear_attack_against_corr_detector}
Suppose the fingerprints $F=\{f_m\}_{m=1}^M$ form an equiangular tight frame.
Then the worst-case type I and type II error probabilities, \eqref{eq:cq:worst_case_error_1} and \eqref{eq:cq:worst_case_error_2}, satisfy
\begin{align*}
\mathrm{P}_\mathrm{I}(F,\tau,\alpha)&\leq Q\bigg[\frac{\gamma}{\sigma}\big(\tau-\mu\big)\bigg],\\
\mathrm{P}_\mathrm{II}(F,\tau,\alpha)&\leq Q\bigg[\frac{\gamma}{\sigma}\Big(\big((1+\mu)\max_{k\in\mathcal{K}}\alpha_k-\mu\big)-\tau\Big)\bigg],
\end{align*}
where $Q(x) := \frac{1}{\sqrt{2\pi}}\int_{x}^{\infty} e^{-u^2/2} du$ and $\mu=\sqrt{\frac{M-N}{N(M-1)}}$.
\end{thm}

\begin{proof}
Under hypothesis $H_0(m)$, the test statistic for the detector \eqref{eq:focused_detection} is 
\begin{eqnarray}
T_z(m) \!\!\!&=&\!\!\! \frac{1}{\gamma^2} \innerprod{\sum_{n \in \mK} \alpha_n f_n + \epsilon}{f_m} \nonumber \\
&=&\!\!\! \sum_{n \in \mK} \alpha_n (\pm \mu) + \epsilon ' , \nonumber
\end{eqnarray}
where $\epsilon '$ is the projection of noise $\epsilon / \gamma$ onto the normed vector $f_m / \gamma$, and, due to symmetry of the variance of $\epsilon$ in all dimensions, $\epsilon ' \sim \mathcal{N}(0,\sigma^2 / \gamma^2)$. Thus, under hypothesis $H_0(m)$, $T_z(m) \sim \mathcal{N}(\sum_{n \in \mK} \alpha_n (\pm \mu),\sigma^2 / \gamma^2)$. We can subtract the mean and divide by the standard deviation to obtain:
\begin{eqnarray} 
\!\!\!\!\text{Prob} \left[  T_m(z) \geq \tau | H_0(m) \right] \!\!\!\!
&=&\!\!\!\!  Q\! \left( \! \frac{\gamma}{\sigma} \! \left[\tau - \! \left(\sum\limits_{n \in \mathcal{K}} \!\! \alpha_n (\pm \mu) \right) \!  \right]\! \right)  \label{eq:err_prob:1a} \nonumber \\
&\leq&\!\!\!\! Q (\frac{\gamma}{\sigma}(\tau- \mu)). \label{eq:err_prob:1b} 
\end{eqnarray}  Note that $Q(x)$ is a decreasing function.  The bound \eqref{eq:err_prob:1b} is obtained by setting all of the coefficients of the coherence to be positive.

Likewise, under hypothesis $H_1(m)$, the test statistic is 
\begin{eqnarray*}
\!\!\!\!T_z(m) \!\!\!&=&\!\!\!\frac{1}{\gamma^2} \! \innerprod{\sum_{n \in \mK} \alpha_n f_n \!+\! \epsilon}{f_m\!\!} \nonumber \\
&=&\!\!\!  \alpha_m\! +\!\!\!\! \sum_{n \in \mK \backslash \{m\}} \!\!\!\! \alpha_n (\pm \mu) \!+\! \epsilon ' .
\end{eqnarray*}
Thus, under hypothesis $H_1(m)$, \[T_z(m) \sim \mathcal{N}\left(\alpha_m +\!\!\!\!\! \sum_{n \in \mK \backslash \{m\}}\!\!\!\!\! \alpha_n (\pm \mu), \frac{\sigma^2}{\gamma^2}\right).\] Since $1 - Q(x) = Q(-x)$, the type II error probability can be bounded as \begin{eqnarray*} \!\text{Prob} \left[  T_m(z) < \tau \right] \!\!\!\!&=&\!\!\!\!  Q\!\! \left( \!\!- \frac{\gamma}{\sigma} \left[\! \tau - \!\!\left( \!\! \alpha_m +\!\!\!\!\! \sum\limits_{n \in \mathcal{K} \backslash \{m\}} \!\!\!\!\! \alpha_n (\pm \mu) \! \right) \! \right] \! \right)  \nonumber \\ 
&\leq&\!\!\!\! Q\! \left(\frac{\gamma}{\sigma}([\alpha_m (1+ \mu) - \mu] - \tau) \right). \end{eqnarray*}

We can now evaluate the worst-case errors \eqref{eq:cq:worst_case_error_1} and \eqref{eq:cq:worst_case_error_2}.  For type I errors, with $m \notin \mK$:
\begin{eqnarray*}
\mathrm{P}_\mathrm{I}(F,\tau,\alpha)\!\!\!\!&=&\!\!\!\!\max_\mathcal{K}\max_{m\not\in\mathcal{K}}\mathrm{P}_\mathrm{I}(F,m,\tau,\mathcal{K},\alpha) \\
&=&\!\!\!\!\max_\mathcal{K}\max_{m\not\in\mathcal{K}} \text{Prob} \left[  T_m(z) \geq \tau | H_0(m) \right] \\
&\leq&\!\!\!\! \max_\mathcal{K}\max_{m\not\in\mathcal{K}} Q (\frac{\gamma}{\sigma}(\tau- \mu)).\label{eq:err_prob:4a} \\
&=&\!\!\!\! Q (\frac{\gamma}{\sigma}(\tau- \mu)).\label{eq:wrst_err:1}
\end{eqnarray*}  For type II errors, \begin{eqnarray*}
\!\!\!\!\!\!\mathrm{P}_\mathrm{II}(F,\tau,\alpha)\!\!\!\!&=&\!\!\!\!\max_\mathcal{K}\min_{m\in \mK} \mathrm{P}_\mathrm{II}(F,m,\tau,\mK,\alpha) \nonumber \\
&=&\!\!\!\!\max_\mathcal{K}\min_{m\in \mK} \text{Prob} \left[  T_m(z) < \tau \right] \\
&\leq&\!\!\!\! \max_\mathcal{K}\min_{m\in \mK} Q(\frac{\gamma}{\sigma}([\alpha_m (1+ \mu) - \mu] - \tau)) \\
&=& \!\!\!\!  Q(\frac{\gamma}{\sigma}([\max_{m \in \mK}\alpha_m (1+ \mu) - \mu] - \tau)). \label{eq:wrst_err:2}
\end{eqnarray*} The last equation follows since once the $\alpha$'s are fixed, the actual coalition $\mK$ does not matter.
\end{proof}

We can further maximize over all possible weightings $\alpha$:
\begin{eqnarray*} \!\!\!\mathrm{P}_\mathrm{I}(F,\tau)\!\!\!\!&=&\!\!\!\!\max_{\alpha} \mathrm{P}_\mathrm{I}(F,\tau,\alpha) \\
&\leq &\!\!\!\! Q (\frac{\gamma}{\sigma}(\tau- \mu)) \\ 
\!\!\!\mathrm{P}_\mathrm{II}(F,\tau) \!\!\!\!\!
&=&\!\!\!\!\max_{\alpha} \mathrm{P}_\mathrm{II}(F,\tau,\alpha) \\
&\leq&\!\!\!\! \max_{\alpha} Q(\frac{\gamma}{\sigma}([\max_{m \in \mK}\alpha_m (1+ \mu) - \mu] - \tau)) \\
&=& \!\!\!\!   Q(\frac{\gamma}{\sigma}([\min_{\alpha}\max_{m \in \mK}\alpha_m (1\!+\! \mu) -\! \mu]\! -\! \tau))   \\
&=& \!\!\!\!   Q(\frac{\gamma}{\sigma}([\frac{1}{K} (1\!+\! \mu) -\! \mu]\! -\! \tau)).  
\end{eqnarray*}  Thus, the vector $\alpha$ which minimizes the value of its maximum element is the uniform weight vector.  This motivates the attackers' use of a uniformly weighted coefficient vector $\alpha$ to maximize the probability that none of the members will be caught.

\subsection{Minimax error analysis}
From a detection standpoint, an important criterion is the minimax error probability \eqref{eq:setup:minmax_err}.  The threshold $\tau$ trades off type I and type II errors.  Thus, there is a value of $\tau$, denoted by $\tau^*$, which minimizes the maximum of the probabilities of the two error types.  Since the bound for the type I error probability is independent of $\alpha$, and the bound for the type II error probability is maximized with a uniform weighting, assume this to be the case. 

\begin{thm}  The minmax probability of error \eqref{eq:setup:minmax_err} can be bounded as:
\begin{eqnarray} \label{eq:minmax_error_bounds}
 Q\bigg( \frac{d^*_{\mathrm{low}}}{2}  \bigg) \leq \mathrm{P}_\mathrm{minmax}(F,\alpha) \leq Q\bigg( \frac{d^*_{\mathrm{up}}}{2} \bigg),
 \end{eqnarray} where \begin{eqnarray*}
 d^*_{\mathrm{low}} \!\!\!&:=&\!\!\! \frac{ \sqrt{ \frac{M}{M-1}} \sqrt{N D_f}   }{ \sigma \sqrt{K (K-1)}} \\
 d^*_{\mathrm{up}} \!\!\!&:=&\!\!\! \!\frac{\sqrt{N D_f}}{\sigma K}\Big(1 - (2K - 1) \mu\Big).
\end{eqnarray*}
\end{thm}  

Note that for orthogonal and simplex fingerprints, the minmax errors are both of the form  \begin{eqnarray*} \mathrm{P}_\mathrm{minmax}(F,\alpha) \!\!\!&=&\!\!\! Q \left( \frac{d^*(K)}{2} \right).  \end{eqnarray*} For orthogonal fingerprints\cite{kiyavash2009regular}, \[ d^*(K) = \frac{\sqrt{ N D_f }}{\sigma K}, \] which is better than $d^*_{\mathrm{low}}(K)$ (the Q function is decreasing).  For simplex fingerprints, $d^*(K)$ is slightly better than both \cite{kiyavash2009regular}: \[ d^*(K) = \frac{\sqrt{ N D_f }}{\sigma K} \frac{M}{M-1}. \]

\begin{proof} The lower bound is the sphere packing lower bound \cite{kiyavash2007sphere}.  For the upper bound, 
\begin{eqnarray*}
&& \!\!\!\!\!\!\!\!\!\!\!\!\!\!\!\!\!\!\!\!\!  P_e(F, \tau, \alpha) = \max \{   P_{I}(F, \tau, \alpha), P_{II}(F, \tau, \alpha) \} \nonumber \\ 
&&\!\!\!\!\!\!\!\!\!\!\!\!\!\! \leq \max \{ Q(\frac{\gamma}{\sigma}(\tau \!-\! \mu)), Q(\frac{\gamma}{\sigma}([\frac{1}{K} (1\!+\! \mu) \!-\! \mu] \! - \!\tau)) \}. 
\end{eqnarray*} Since the test statistic $T_z(m)$ is normally distributed with the same variance under either of the hypotheses $H_0(m)$ and $H_1(m)$, the value of $\tau$ that minimizes this upper bound is the average of the means $\mu$ and $\frac{1+\mu}{K}-\mu$, namely $\smash{\tau^*:=\frac{1+\mu}{2K}}$. Using this $\tau^*$ and recalling \eqref{eq:energy}, we have \begin{eqnarray*}
\mathrm{P}_\mathrm{minmax}(F, \alpha)\!\!\! &=&\!\!\!\min_{\tau} \mathrm{P}_\mathrm{e}(F,\tau,\alpha) =  \mathrm{P}_\mathrm{e}(F,\tau^*,\alpha) \\
&\leq& \!\!\!Q\bigg(\frac{\gamma}{\sigma}\big(\tau^*-\mu\big)\!\bigg)  \\
&=& \!\!\!Q\bigg(\!\frac{\sqrt{N D_f}}{2\sigma K}\Big(1 - (2K - 1) \mu\Big)\!\bigg)  \label{eq.minimax.upbnderr} \\
&=&\!\!\! Q( d^*_{\mathrm{up}}(K) /2). \end{eqnarray*} \end{proof}

Consider the regime where $N$ is large, $M$ grows linearly or faster than $N$, and $\WNR$ is constant (in particular $0$, so $D_f = \sigma^2$).  Then $\mu \approx 1/\sqrt{N}$,  
\begin{eqnarray*}
 \quad d^*_{\mathrm{low}} \approx \frac{ \sqrt{ N }}{ K  }, \quad \text{and} \quad d^*_{\mathrm{up}}
 \approx \frac{ \sqrt{ N }}{ K  } \!\! \left(\!\! 1\! - \frac{2 K}{\sqrt{N}} \! \right)=  \frac{ \sqrt{ N }}{ K  } - 2.
\end{eqnarray*}  If $K \ll \sqrt{N}$, then both bounds go to infinity so $\mathrm{P}_\mathrm{minmax}(F,\alpha) \to 0$.  Also, the geometric figure of merit  \eqref{eqn:dist_guilt_notguilt} then behaves as 
\begin{eqnarray*}
\udmk \approx \frac{1}{K} \approx \sqrt{N} d^*_{\mathrm{up}}.
\end{eqnarray*}  If $K$ is proportional to $\sqrt{N}$, then $\mathrm{P}_\mathrm{minmax}(F,\alpha)$ is bounded away from 0.

We can also compute the error exponent for this test: \[\mathrm{e}(F,\tau^*,\alpha):= -\!\!\lim_{N\rightarrow\infty}\frac{1}{N}\ln\mathrm{P}_\mathrm{e}(F,\tau^*,\alpha).\] 
\begin{cor} If $M \gg N \gg K^2$, then the error exponent is
\begin{eqnarray}
\mathrm{e}(F,\tau^*,\alpha) = \frac{1}{8 K^2}. \label{eq.error.exponent}
\end{eqnarray} \end{cor}
\begin{proof}
The proof follows by applying the asymptotic equality $\ln Q(t) \sim -\frac{t^2}{2}$ as $t\to \infty$ to the bounds in \eqref{eq:minmax_error_bounds}.  The bounds are asymptotically equivalent.
\end{proof} Note that this error exponent is the same as in the simplex case \cite{kiyavash2009regular}.  As $K \to \infty$, the error exponent \eqref{eq.error.exponent} goes to zero.

\newcommand{\Pd}{\mathrm{P}_{\mathrm{d}}}
\newcommand{\Pfa}{\mathrm{P}_{\mathrm{fa}}}

\section{Simulations} \label{sec:simulations}

In Section~\ref{sec:error_analysis}, the worst-case error probabilities were analyzed, where the worst case was over all collusions.  Here we investigate average case behavior.  We examine the probability of detecting at least one guilty user, $\mathrm{P}_\mathrm{d}$ \eqref{eq:cq:P_d}, as a function of $K$ with the false-alarm $
\mathrm{P}_\mathrm{fa}$ \eqref{eq:cq:P_fa} fixed below a threshold.  The threshold can be interpreted as the (legally) allowable limit for probability of false accusation.  We compare the average case performance of ETF fingerprints, simplex fingerprints \cite{kiyavash2009regular}, and orthogonal fingerprints \cite{cox1997secure} for four dimension sizes $N \in \{195, 651, 2667, 8128\}$.  The ETF construction described in Example~\ref{ex:ETF} was used.  The results demonstrate that ETF fingerprints perform almost as well as both orthogonal and simplex fingerprints, while accommodating several times as many users.  We now describe the design of the simulations.

\subsection{Design}

For a fixed signal dimension size $N$, ETF, orthogonal, and simplex fingerprints were created.  The ETF fingerprint design was constructed using the method shown in Example~\ref{ex:ETF}. For $N = 195$, a $(2,7,91)$-Steiner system was used \cite{ccrwest}, yielding $M=1456$ fingerprints.  For $N=651$ and $N=2667$, the Steiner systems were constructed using projective geometry, with $M=2016$ and $M=8128$ respectively \cite{fickus2010steiner}.  For $N = 8128$, a $(2,2,2^7)$-Steiner system was used, giving $M = 16,384$ fingerprints \cite{fickus2010steiner}.

The orthogonal fingerprint design was constructed using an identity matrix.  For actual embedding, the orthogonal fingerprints should be randomly rotated to ensure difficulty in removing them.  One method to achieve this is to randomly generate i.i.d. Gaussian vectors in $N$ dimensional space, orthogonalize them, then scale them to have the same energy \cite{cox1997secure}.  For detection purposes, however, rotations of the basis vectors are inconsequential.

The simplex fingerprint design was constructed using the following method of sequentially fixing the values of the fingerprints in each dimension.  The vectors of a regular simplex are equidistant from the center, having the same power, and have inner products equal to $-\frac{1}{N}$ \cite{munkres1984elements}.  Letting the first vector have a one in the first row and zeros in the other rows forces every other vector's first element to be $-\frac{1}{N}$ to satisfy the inner product constraint. For the second vector, choose its second element to satisfy the power constraint and set the elements of the remaining rows to zero.  Using the inner product constraint, we can find the value of the second element for all other vectors.  Repeating these steps constructs a regular simplex.


\begin{figure*}
\centering
\begin{tabular}{cc}
\begin{minipage}{\columnwidth}
\includegraphics[height=.3\textheight]{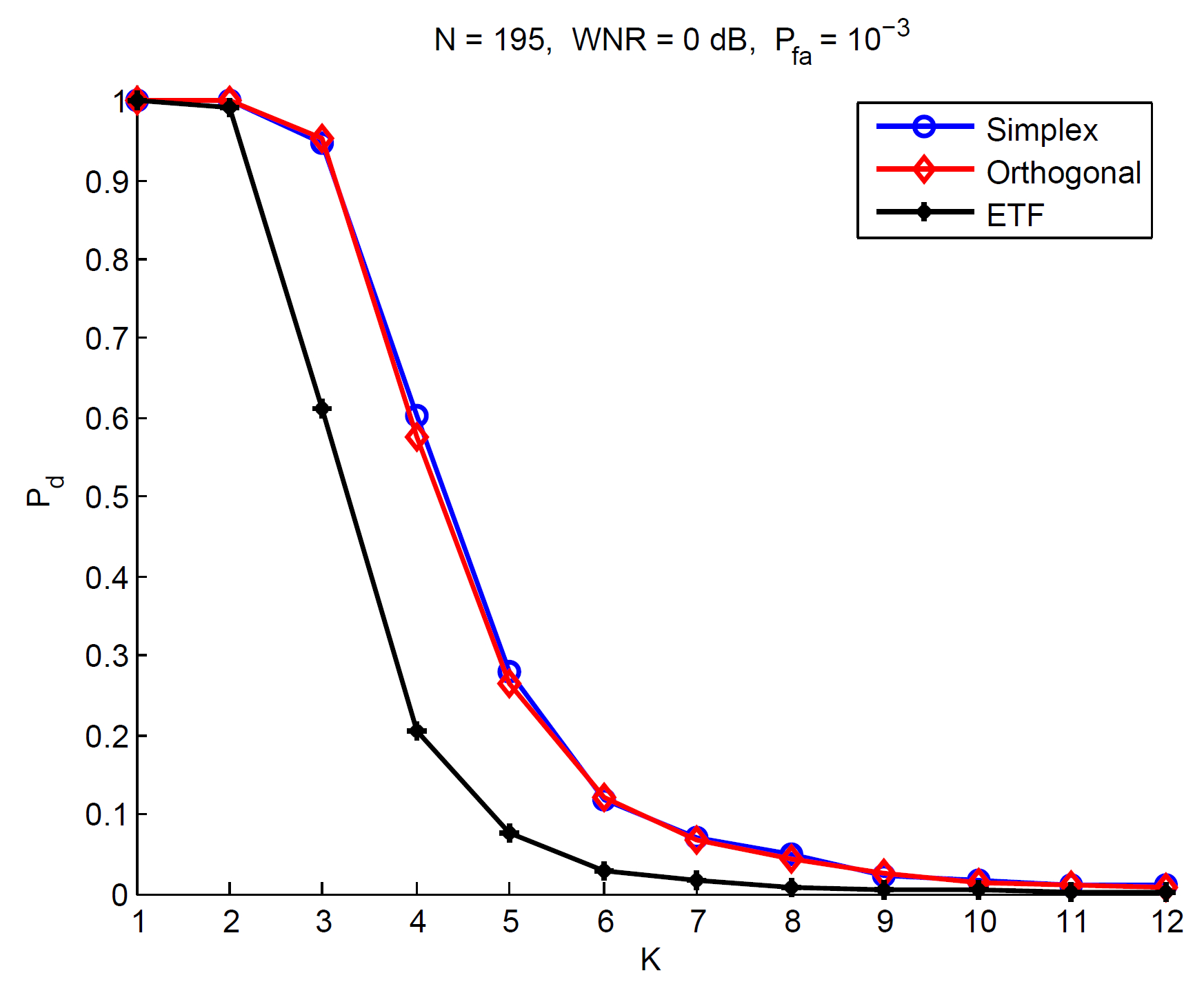}
\caption{\label{fig:sim:N_195} A plot of the probability of detecting at least one colluder ($\Pd$) as a function of the number of colluders ($K$).  The threshold $\tau$ is picked to be the minimum threshold to fix $\Pfa \leq 10^{-3}$.  The WNR is $0$ dB and $N = 195$.  The (maximum) number of fingerprints were $195$ for the orthogonal, $196$ for the simplex, and $1456$ for the ETF construction.}
\end{minipage}
&
\begin{minipage}{\columnwidth}
\centering
\includegraphics[height=.3\textheight]{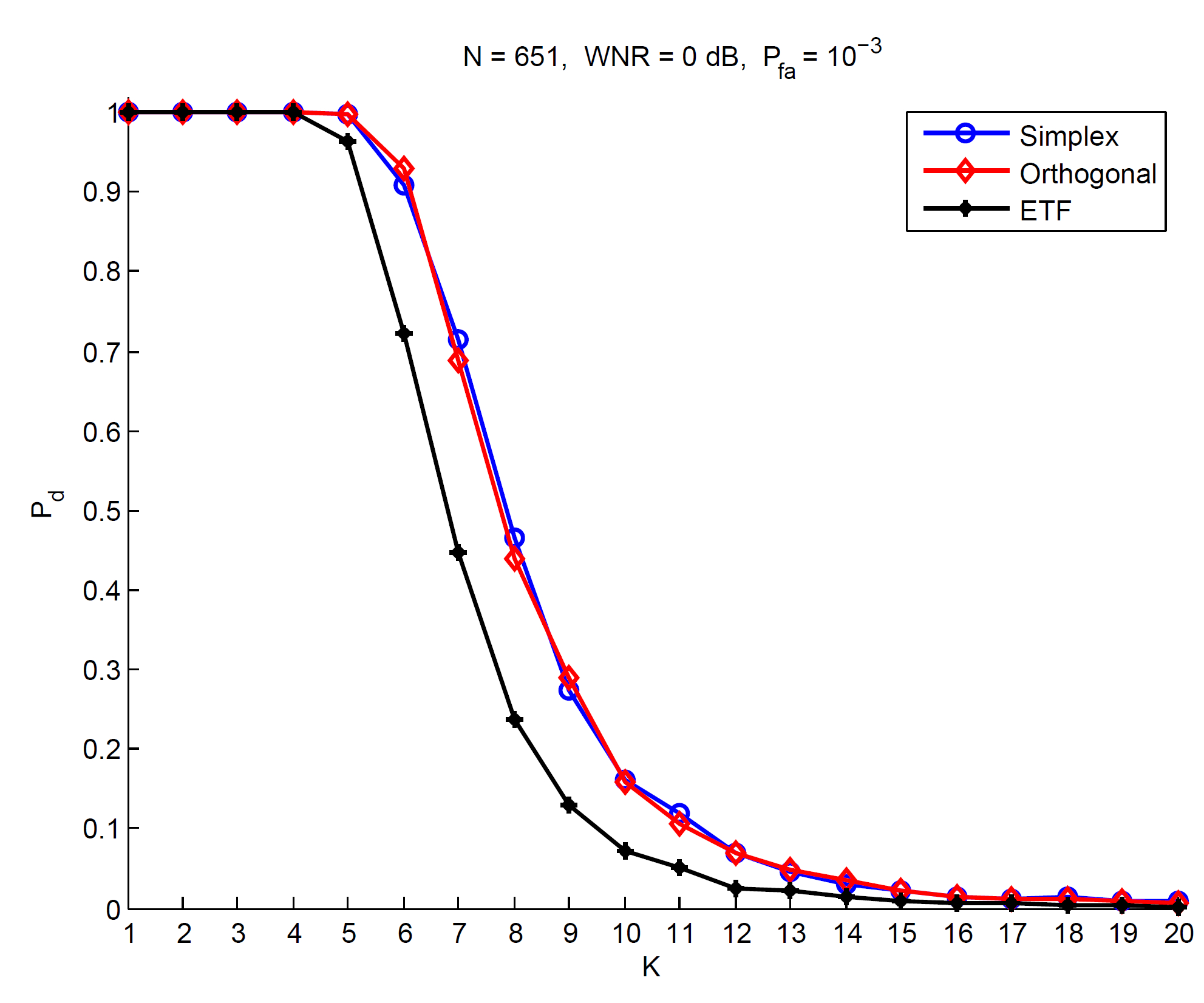}
\caption{\label{fig:sim:N_651} A plot of the probability of detecting at least one colluder ($\Pd$) as a function of the number of colluders ($K$).  The threshold $\tau$ is picked to be the minimum threshold to fix $\Pfa \leq 10^{-3}$.  The WNR is $0$ dB and $N = 651$. The (maximum) number of fingerprints were $651$ for the orthogonal, $652$ for the simplex, and $2016$ for the ETF construction.}
\end{minipage} \\ \\
\begin{minipage}{\columnwidth}
\includegraphics[height=.3\textheight]{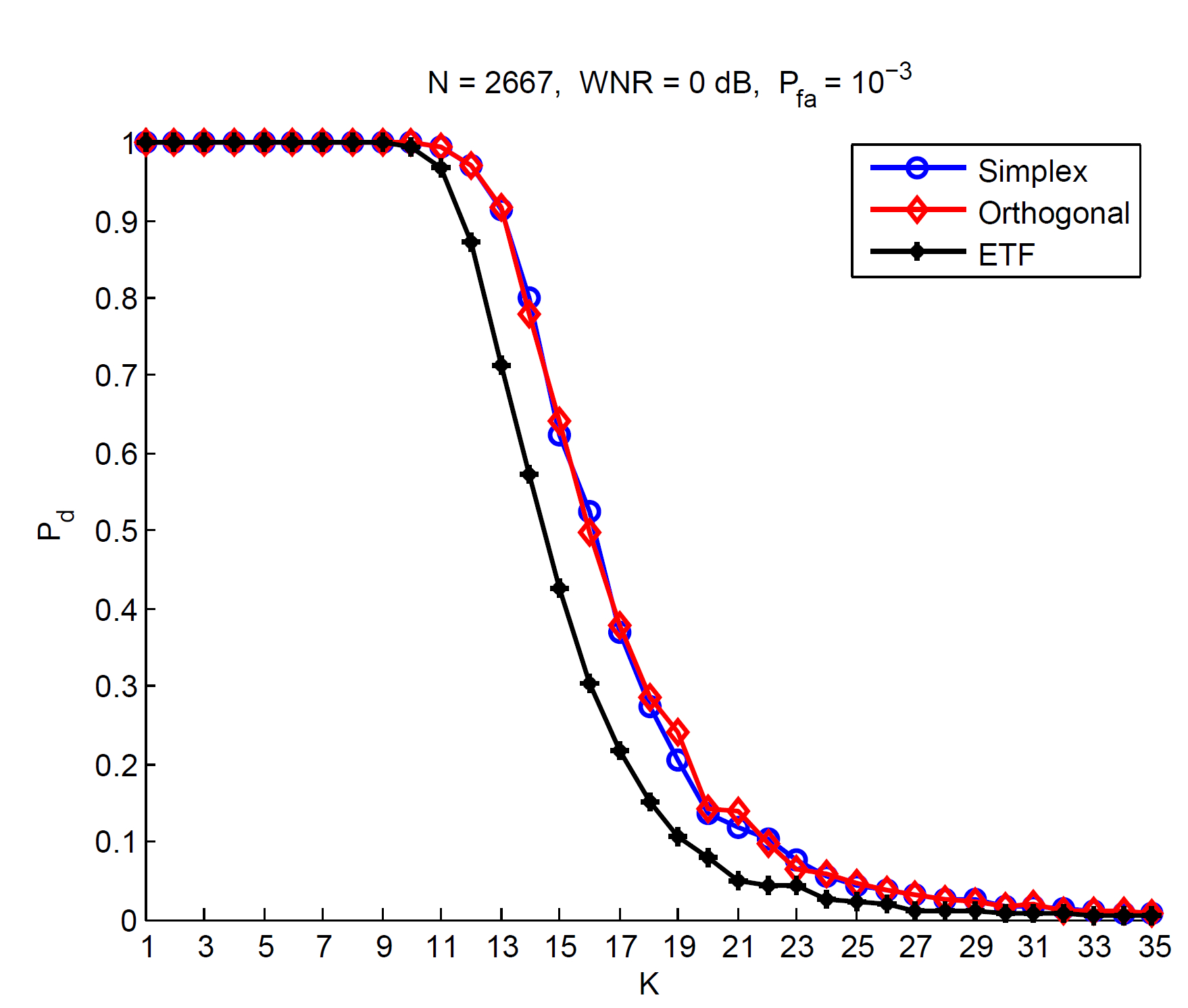}
\caption{\label{fig:sim:N_2667} A plot of the probability of detecting at least one colluder ($\Pd$) as a function of the number of colluders ($K$).  The threshold $\tau$ is picked to be the minimum threshold to fix $\Pfa \leq 10^{-3}$.  The WNR is $0$ dB and $N = 2667$. The (maximum) number of fingerprints were $2667$ for the orthogonal, $2668$ for the simplex, and $8128$ for the ETF construction.}
\end{minipage}
&
\begin{minipage}{\columnwidth}
\includegraphics[height=.3\textheight]{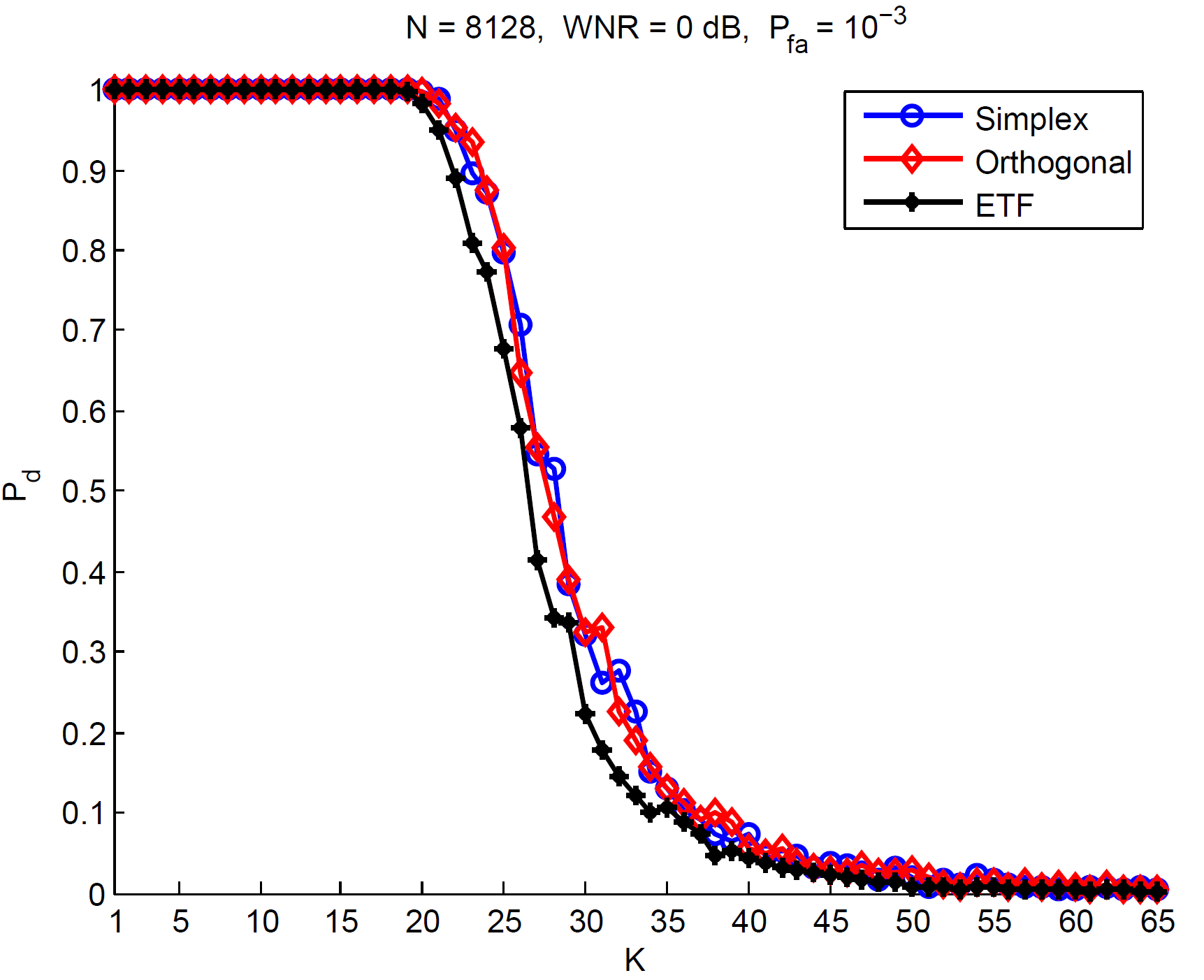}
\caption{\label{fig:sim:N_8128} A plot of the probability of detecting at least one colluder ($\Pd$) as a function of the number of colluders ($K$).  The threshold $\tau$ is picked to be the minimum threshold to fix $\Pfa \leq 10^{-3}$.  The WNR is $0$ dB and $N = 8128$.  The (maximum) number of fingerprints were $8128$ for the orthogonal, $8129$ for the simplex, and $16,384$ for the ETF construction.}
\end{minipage}
\end{tabular}
\end{figure*}

Linear collusion attacks were simulated separately for the different designs and collusion sizes $K$.  For each attack, $K$ of the $M$ fingerprints were randomly chosen and uniformly averaged.  Next, an i.i.d. Gaussian noise vector was added with per-sample noise power $\sigma^2 = D_f$, corresponding to a $\WNR$ \eqref{eq:wnr} of 0 dB \cite{wang2005anti}.  

The test statistics $T_z(m)$ \eqref{eq:test_stat} were then computed for each user $m$.  For each threshold $\tau$, it was determined whether there was a detection event (at least one colluder with $T_z(m) > \tau$) and/or a false alarm (at least one innocent user with $T_z(m) > \tau$).  In total, 50,000 attacks were simulated, and the detection and false alarm counts were averaged.  Then the minimal $\tau$ value was selected for which $\mathrm{P}_{\mathrm{fa}} \leq 10^{-3}$.  This induced the corresponding $\mathrm{P}_{\mathrm{d}}$.

\subsection{Results}

We ran experiments with four different dimension sizes $N \in \{195, 651, 2667, 8128\}$.    The noise level was kept at $\mathrm{WNR} = 0 \text{ dB}$.  The value of $K$ varied between $1$ and sufficiently large values so  $\Pd$ approached zero.  Plots for the probability of detection $\Pd$ as a function of the size of the coalition $K$ are shown in Figures~\ref{fig:sim:N_195}--\ref{fig:sim:N_8128} for $N \in \{195, 651, 2667, 8128\}$ respectively.  The largest values of $K$ for which at least one attacker can be caught with probability (nearly) one under the $\Pfa$ constraint are about 2, 5, 11, and 21 respectively.  Overall, ETF fingerprints perform comparably to orthogonal and simplex fingerprints while accommodating many more users.

\ifCLASSOPTIONcaptionsoff
  \newpage
\fi



\bibliographystyle{IEEEtran}



%





\end{document}